\documentclass[runningheads]{llncs}
\usepackage[T1]{fontenc}
\usepackage{graphicx}
\usepackage{amsmath, amsthm, amssymb}
\usepackage[colorlinks=true, allcolors=blue]{hyperref}
\usepackage{cleveref}
\usepackage{xspace}
\usepackage{todonotes}
\usepackage{comment}
\usepackage{algpseudocode}
\usepackage{enumitem}
\usepackage{booktabs,tabularx}
\usepackage{lipsum}

\theoremstyle{plain}

\newtheorem{obs}[theorem]{Observation}
\theoremstyle{definition}
\newtheorem{defn}{Definition}[section]

\newtheorem{prb}{Problem}

\newcommand{\homo}{\textsc{Hom}$_<$\xspace}

\newcommand{\NP}{\textsf{NP}\xspace}
\newcommand{\wone}{\textsf{W[1]}\xspace}
\newcommand{\XP}{\textsf{XP}\xspace}
\newcommand{\Oh}{\mathcal{O}}
\newcommand{\ret}{\textsc{RET}$_<$}
\newcommand{\core}{\textsc{CORE}$_<$\xspace}
\newcommand{\slice}{\textsc{SLICE}$_{<gh}$}

\newcommand{\allsub}{\textsc{ALLSUB}$_<$}

\newcommand{\sub}{\textsc{SUB}$_{<tu}$}
\newcommand{\PP}{\textsf{P}}
\newcommand{\twosat}{2-\textsf{SAT}}
\newcommand{\corek}{\textsc{CORE}$_<^k$\xspace}
\newcommand{\corechi}{\textsc{CORE}$_<^{\chi^<(G)}$\xspace}
\newcommand{\NN}{\mathbb{N}}

\newcommand{\problemStatement}[3]{%
  \begin{center}
  \begin{tabularx}{\columnwidth}{@{}lX@{}}
  \toprule
  \multicolumn{2}{@{}c@{}}{\textsc{#1}}\tabularnewline
  \midrule
  \bfseries Input:    & #2 \\
  \bfseries Question: & #3 \\
  \bottomrule
  \end{tabularx}
  \end{center}
}

\begin{document}
\title{On Computational Aspects of Cores of Ordered Graphs}
%
%
\author{Michal \v{C}ert\'{\i}k \inst{1}\orcidID{0009-0008-6880-4896} \and
Andreas Emil Feldmann \inst{2}\orcidID{0000-0001-6229-5332} \and
Jaroslav Ne\v{s}et\v{r}il \inst{1}\orcidID{0000-0002-5133-5586} \and
Pawe\l{} Rz\k{a}\.zewski \inst{3}\orcidID{0000-0001-7696-3848}}
\authorrunning{M. \v{C}ert\'{\i}k, A. E. Feldmann, J. Ne\v{s}et\v{r}il, P. Rz\k{a}\.zewski}
%
\institute{Computer Science Institute, Faculty of Mathematics and Physics \\
Charles University\\
Prague, Czech Republic \and
Department of Computer Science, University of Sheffield\\
Sheffield, United Kingdom \and
Warsaw University of Technology and University of Warsaw\\
Warsaw, Poland}
\maketitle              
\begin{abstract}

An ordered graph is a graph enhanced with a linear order on the vertex set.
An ordered graph is a \emph{core} if it does not have an order-preserving homomorphism to a proper subgraph. We say that $H$ is the core of $G$ if (i) $H$ is a core, (ii) $H$ is a subgraph of $G$, and (iii) $G$ admits an order-preserving homomorphism to $H$.

We study complexity aspects of several problems related to the cores of ordered graphs.
Interestingly, they exhibit a different behavior than their unordered counterparts.

We show that the \emph{retraction} problem, i.e., deciding whether a given graph admits an ordered-preserving homomorphism to its specific subgraph, can be solved in polynomial time.
On the other hand, it is \NP-hard to decide whether a given ordered graph is a core.
In fact, we show that it is even \NP-hard to distinguish graphs $G$ whose core is largest possible (i.e., if $G$ is a core) from those, whose core is the smallest possible, i.e., its size is equal to the ordered chromatic number of $G$.
The problem is even \wone-hard with respect to the latter parameter.

\end{abstract}

\section{Introduction}

An \emph{ordered graph} is a graph whose vertex set is totally ordered (see example in Figure~\ref{fig:OrdHomsInterval}).

For two ordered graphs $G$ and $H$, an \emph{ordered homomorphism} from $G$ to $H$ is a mapping $f$ from $V(G)$ to $V(H)$ that preserves edges and the orderings of vertices, i.e.,
\begin{enumerate}
    \item for every $uv \in E(G)$ we have $f(u)f(v) \in E(H)$,
    \item for $u,v \in V(G)$, if $u \leq v$, then $f(u) \leq f(v)$.
\end{enumerate}

Note that the second condition implies that the preimage of every vertex of $H$ forms a segment (or \emph{interval}) in the ordering of $V(G)$ (see Figure~\ref{fig:OrdHomsInterval}).

\begin{figure}[ht]
\begin{center}
\includegraphics[scale=0.8]{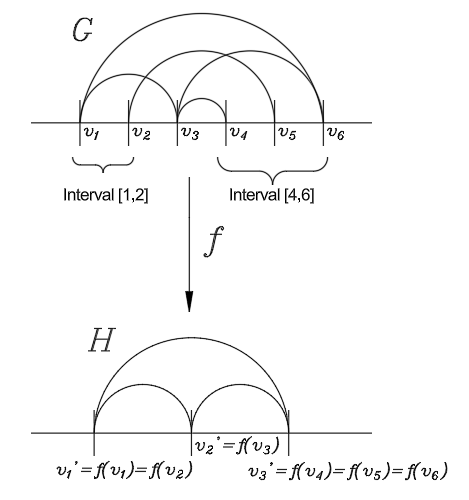}
\end{center}
\caption{Ordered Homomorphism $f$ and Independent Intervals.}
\label{fig:OrdHomsInterval}
\end{figure}

Similarly to unordered graphs (see, e.g. ~\cite{HellNesetrilGraphHomomorphisms}), we can also consider a problem of minimum coloring of $G$ and define \emph{(ordered) chromatic number} $\chi^<(G)$ to be the minimum $k$ such that $V(G)$ can be partitioned into $k$ disjoint independent intervals. The ordered chromatic number is sometimes called \emph{the interval chromatic number} (see ~\cite{DBLP:conf/bcc/Tardos19}).

Now suppose that an ordered graph $H$ is a subgraph of an ordered graph $G$. Let us first start with a definition of an \emph{Ordered Retraction} of $G$ to $H$ being an ordered homomorphism $f:G\to H$ such that $f(v)=v$ for all $v\in V(H)$. If this ordered retraction of $G$ to $H$ exists, we say that $G$ \emph{order-retracts} to $H$ or that $H$ is an \emph{order-retract} of $G$. Notice that, as for digraphs (see \cite{HellNesetrilGraphHomomorphisms}), if ordered retraction exists, we have ordered homomorphism of $G$ to $H$ and inclusion-ordered homomorphism of $H$ to $G$, therefore $G$ and $H$ are (ordered) homomorphically equivalent. Note that this also implies that if $H$ is an ordered core of an ordered graph $G$, then $\chi^<(G)=\chi^<(H)$.

We then define \emph{an Ordered Core} as an ordered graph that does not retract to a proper ordered subgraph.

Lastly, we also define an \emph{Ordered Matching} as an ordered graph $G$ where each vertex has exactly one edge incident to it, and we define an (ordered) \emph{Directed Path} $P_m$ as an ordered graph with $m$ vertices and edges
$E(P_m)=\{v_iv_{i+1}|i=1,\ldots,m-1;v_i\in V(P_m)\}$. Note that ordered homomorphism is injective on a directed path and the image of directed path contains a spanning directed path.


\section{Motivation}

Ordered graphs emerge naturally in several contexts: Ramsey theory (~\cite{Nesetril1996}, ~\cite{Hedrlín1967}, ~\cite{balko2022offdiagonal}), extremal graph theory (~\cite{Pach2006}, ~\cite{conlon2016ordered}), category theory (~\cite{Hedrlín1967}, ~\cite{nesetril2016characterization}), among others. Recently, the concept of twin-width in graphs has been demonstrated to be equivalent to NIP (or dependent) classes of ordered graphs (~\cite{bonnet2021twinwidth}, also refer to ~\cite{bonnet2024twinwidth}), creating a link between graph theory and model theory. 

Homomorphisms of ordered graphs both verify and enrich this research avenue: they impose more constraints (compared to standard homomorphisms, see, e.g., ~\cite{HellNesetrilGraphHomomorphisms}) yet exhibit their own complexity (see, e.g., ~\cite{Axenovich2016ChromaticNO}). 


The study of ordered graphs exhibits a remarkable breadth, extending far beyond its foundational theoretical challenges and structural questions. Work involving ordered structures appears across an impressive cross-section of scientific and technological disciplines. Applications arise in physics \cite{verbytskyi2020hepmc3-80a}, the life sciences \cite{goerttler2024machine-7eb}, and modern AI methodologies, including large language models \cite{ge2024can-da2}, neural networks \cite{guo2019seq2dfunc-c64}, machine learning \cite{goerttler2024machine-7eb}, and self-supervised learning \cite{LimOrderLearning2020}. Connections also emerge in data analysis and subspace clustering \cite{xing2025block-diagonal-073}, systems and networking research \cite{li2015understanding-1cc}, software optimization \cite{romansky2020approach-60b}, and various aspects of computer security, such as malware detection \cite{thomas2023intelligent-c25} and dynamic system-call sandboxing \cite{zhang2023building-6de}. Additional uses appear in business process management \cite{kourani2023business-759}, workflow modeling \cite{kourani2023scalable-a0f}, decision-making frameworks \cite{wang2022improved-f6a}, fault-tolerant distributed systems \cite{chen2023pgs-bft-dcb}, blockchain protocol design \cite{malkhi2023bbca-chain-bfa}, and even educational curriculum development \cite{kuzmina2020curriculum-cdf}. Ordered graph perspectives further influence studies of multi-linear forms \cite{bhowmik2023multi-linear-2f5}, graph grammars \cite{brandenburg2005graph-grammars-2fe}, rigidity and reconstruction problems \cite{connelly2024reconstruction-912}, combinatorial objects such as shuffle squares \cite{grytczuk2025shuffle-a46}, and various tiling theories \cite{balogh2022tilings-d8d}. Together, these examples highlight the extensive and diverse reach of ordered graph theory across contemporary research areas.

Studying the core of an ordered graph is essential in studying the ordered homomorphisms, since it gives a canonical representative (of the smallest in its class) for graph homomorphisms (see, e.g., ~\cite{HELL1992117}, ~\cite{HellNesetrilGraphHomomorphisms}). We show in Theorem ~\ref{thm:Unique}, that every finite ordered graph is homomorphically equivalent to a unique ordered core (up to the isomorphism). 
This “minimal witness” preserves homomorphism-invariant properties (e.g., colorability as an ordered homomorphism to $K_k$) while stripping away inessential structure, so arguments about density, dualities or chromatic bounds can be carried out on a simpler object without loss of generality (see, e.g., ~\cite{nescer2023duality}). In practice, cores organize graphs into homomorphism equivalence classes, clarify when two graphs encode the same constraint language, and provide a clean target for reductions.

In this way, cores could also be seen as a bridge between combinatorics and computation. In the modern view of constraint satisfaction problems (CSPs), many complexity classifications reduce to structural and algebraic properties of target templates, where passing to a core reveals the essential structure that governs tractability vs. hardness (see, e.g., ~\cite{bodirsky2025graph-5a2}). These perspectives extend beyond unordered (or ordered) graphs to relational structures (including ordered variants, see, e.g., ~\cite{nescerfelrza2023Systems}), so “core-ification” becomes a standard normalization step in both theory and algorithms. Understanding the complexity of these exercises, therefore, becomes crucial in these endeavors.

\section{Ordered Core}
\label{sec:core}

We start with the following statement, analogous to unordered digraphs (see \cite{HellNesetrilGraphHomomorphisms}), providing an equivalent definition for ordered core and its uniqueness per each ordered graph.

\begin{theorem}
\label{thm:Unique}
Let $G$ be an ordered graph. Then an ordered graph $G$ is a core if and only if there is no ordered homomorphism from $G$ to a proper ordered subgraph of $G$. Every ordered graph is homomorphically equivalent to a unique ordered core.
\end{theorem}

\begin{proof}
To prove the left direction, if an ordered graph $G$ retracts to a proper ordered subgraph $H$ then there is an ordered homomorphism $f:G\to H$. Conversely, if there is an ordered homomorphism $f:G\to H$, where $H$ is a proper ordered subgraph of $G$, then let $H$ be a proper ordered subgraph of $G$ with fewest vertices to which $G$ is homomorphic. Then there cannot be an ordered homomorphism of $H$ to its proper subgraph, and therefore any of $H$ to $H$ must be an automorphism of $H$. Let now $f:G\to H$ be an ordered homomorphism. Thus, the restriction $a$ of $f$ to $H$ is an automorphism of $H$ and therefore has an inverse automorphism $a^{-1}$. Then $a^{-1}\circ f$ is an ordered retraction of $G$ to $H$.

To prove the second part of the statement, we notice that if ordered subgraphs $H$ and $H^{'}$ are retracts of $G$ then there exist ordered homomorphisms $f:H\to H^{'}$ and $g:H^{'}\to H$ (from transitivity). If both $H$ and $H^{'}$ are also cores, then $f \circ g$ and $g \circ f$ must be automorphisms (as shown above); hence $H$ and $H^{'}$ must be isomorphic.
\end{proof}

In this article, we will therefore consider a computational problem that we will call \core{}.

\label{Prb:Core}
\begin{prb}
\end{prb}
\problemStatement{\core{}}
  {Ordered graph $G$.}
  {Is there a non-surjective homomorphism $G\to G$?}

Note that, as in \cite{HELL1992117}, the cores are precisely the ordered graphs for which the answer is NO. This formulation has the advantage of giving a problem that belongs to the class \NP, since it is easy to verify in polynomial time that a proposed mapping is indeed a non-surjective ordered homomorphism (and it is a decision problem).

\section{Polynomial Ordered Core Problems}

As a first step, consider the following problem, which we call \ret{}.

\label{Prb:Core}

\begin{prb}
\end{prb}
\problemStatement{\ret{}}
  {Ordered graphs $G$ and a subset $X\in V(G)$.}
  {Does there exist a function $f:G\to X$, which is identity on $X$, and is an ordered homomorphism from $G$ to $G[X]$?}

\begin{theorem}
\label{thm:RetInPP}
\ret{} is in \PP.
\end{theorem}

\begin{proof}
We prove this statement by reduction to an instance of \twosat. This problem is defined by its input Boolean formula $\varphi$ in conjunctive normal form (CNF), where each clause has at most two literals, and we ask if there exists a truth assignment to the variables that makes $\varphi$ true. Formally,

$$
\twosat=\{\varphi\in CNF| \text{ each clause has } \le 2 \text{ literals and $\varphi$ is satisfiable}\}.
$$

Unlike general SAT, \twosat ~can be solved in polynomial time. The classical result is due to ~\cite{even1976complexity-b05}, who even showed that \twosat ~is solvable in linear time.

First, we denote an ordered graph $H$ as an induced ordered subgraph of $G$, induced by a subset $X\in V(G)$.

By denoting $n=|G|$ and $h=|H|$, we can separate the vertices of $G$ into two sets.

\begin{enumerate}
    \item With $j\in [h]$, the vertices $v_j$ corresponding to $V(G[X])$.
    \item With $i\in [n-h]$, the vertices $x_i$ corresponding to $V(G)-V(G[X])$.
\end{enumerate}

We will then define sets of vertices $X_k, k\in [h+1]$, where $X_k$ contains all the vertices $x_i$ that are in between $v_{k-1}$ and $v_k$. We notice that there will be maximum $h+1$ nonempty sets $X_k$. Of course, $V(G)-V(G[X])=\bigcup_{k \in [h+1]} X_k$ and $V(G)$ will look as follows:

\begin{align*}
X_1,v_1,X_2,v_2,\ldots,X_h,v_h,X_{h+1}.
\end{align*}

Note that for any $k\in [h+1]$, $X_k$ can be an empty set.
Furthermore, we shall denote $x^l_k, k\in [h+1], l\in [|X_k|]$ as one of the vertices of $V(G)-V(G[X])$ that belongs to $X_k$ and is the $l$-th vertex of $X_k$ (following the order of $G$).

As a first observation, we note that each vertex $x^l_k\in (V(G)-V(G[X]))$ has two options:

\begin{enumerate}
    \item Either it will map to $v_{k-1}\in H$,
    \item or it will map to $v_{k} \in H$.
\end{enumerate}

Notice that the cases of $x^l_1$ and $x^l_{h+1}$ will be only a simplification as they will map to $v_1\in H$ and $v_h\in H$, respectively.

We will therefore establish that for each vertex $x^l_k\in (V(G)-V(G[X]))$ we have a variable $s_i, i\in [n-h]$, which is false if $x^l_k$ maps to $v_{k-1}\in H$ and true if $x^l_k$ maps to $v_{k}\in H$. Similarly to $x_i$, we will denote this variable $s^l_k,k\in [h+1], l\in [|X_k|]$.

We then have the following constraints in our \ret{} problem, and the way we enforce them in our instance of \twosat:

\begin{enumerate}
    \item Let $x^{l_1}_k, x^{l_2}_k\in (V(G)-V(G[X])))$ and $x^{l_1}_k<x^{l_2}_k$. Then, if $x^{l_2}_k$ maps to $v_{k-1}\in V(H)$, $x^{l_1}_k$ must also map to $v_{k-1}$.
    
    We enforce this constraint in our \twosat ~by a clause $(s^{l_2}_k \lor \neg s^{l_1}_k)$ which is equivalent to $(\neg s^{l_2}_k \rightarrow \neg s^{l_1}_k)$. We see that this way if $s^{l_2}_k$ is false and $s^{l_1}_k$ is true, the forbidden combination, the clause is false. We see that by going through all the pairs of vertices, we create at most $\sum_{k=1}^{h+1} \binom{|X_k|}{2}$ clauses.
    
    \item Similarly for the opposite constraint — let $x^{l_1}_k, x^{l_2}_k\in (V(G)-V(G[X]))$ and $x^{l_1}_k<x^{l_2}_k$. Then, if $x^{l_1}_k$ maps to $v_{k}\in V(H)$, $x^{l_2}_k$ must also map to $v_{k}$.
    
    We enforce this constraint in our \twosat by a clause $(\neg s^{l_1}_k \lor s^{l_2}_k)$ which is equivalent to $( s^{l_1}_k \rightarrow s^{l_2}_k)$. We see that this way if $s^{l_1}_k$ is true and $s^{l_2}_k$ is false, the forbidden combination, the clause is false. 
    
    Again, we see that by going through all the pairs of vertices, we create at most $\sum_{k=1}^{h+1} \binom{|X_k|}{2}$ clauses.
    
    \item Let $x^{l_1}_{k_1}, x^{l_2}_{k_2}\in (V(G)-V(G[X])), l_1\in [|X_{k_1}|], l_2\in [|X_{k_2}|], k_1,k_2\in [h+1]$ be vertices in $G$ in between which there is an edge. Then we check four different mappings of $x^{l_1}_{k_1}, x^{l_2}_{k_2}\in V(G)$ to $H$:
    
\begin{align*}
    x^{l_1}_{k_1}\to v_{k_1 - 1} &\text{ and }  x^{l_2}_{k_2}\to v_{k_2 - 1}\\
    x^{l_1}_{k_1}\to v_{k_1} &\text{ and }  x^{l_2}_{k_2}\to v_{k_2 - 1}\\
    x^{l_1}_{k_1}\to v_{k_1 - 1} &\text{ and }  x^{l_2}_{k_2}\to v_{k_2}\\
    x^{l_1}_{k_1}\to v_{k_1} &\text{ and }  x^{l_2}_{k_2}\to v_{k_2}\\
\end{align*}
    
    For each of these cases, if there does not exist an edge $v_i v_j$ in $H$, we create a following restrictive clause in our \twosat:
    
\begin{itemize}
    \item For the case where $x^{l_1}_{k_1}\to v_{k_1 - 1}$ and $x^{l_2}_{k_2}\to v_{k_2 - 1}$, if there is no edge $v_{k_1 - 1} v_{k_2 - 1}$ in $H$, then we create the clause $( s^{l_1}_{k_1} \lor s^{l_2}_{k_2})$. We see that this way, if $s^{l_1}_{k_1}$ is false and $s^{l_2}_{k_2}$ is false, the forbidden combination, the clause is false.
    \item For the case where $x^{l_1}_{k_1}\to v_{k_1}$ and $x^{l_2}_{k_2}\to v_{k_2 - 1}$, if there is no edge $v_{k_1} v_{k_2 - 1}$ in $H$, then we create the clause $(\neg s^{l_1}_{k_1} \lor s^{l_2}_{k_2})$. We see that this way, if $s^{l_1}_{k_1}$ is true and $s^{l_2}_{k_2}$ is false, the forbidden combination, the clause is false.
    \item For the case where $x^{l_1}_{k_1}\to v_{k_1 -1}$ and $x^{l_2}_{k_2}\to v_{k_2 }$, if there is no edge $v_{k_1 - 1} v_{k_2}$ in $H$, then we create the clause $( s^{l_1}_{k_1} \lor \neg s^{l_2}_{k_2})$. We see that this way, if $s^{l_1}_{k_1}$ is false and $s^{l_2}_{k_2}$ is true, the forbidden combination, the clause is false.
    \item For the case where $x^{l_1}_{k_1}\to v_{k_1}$ and $x^{l_2}_{k_2}\to v_{k_2}$, if there is no edge $v_{k_1} v_{k_2}$ in $H$, then we create the clause $(\neg s^{l_1}_{k_1} \lor \neg s^{l_2}_{k_2})$. We see that this way, if $s^{l_1}_{k_1}$ is true and $s^{l_2}_{k_2}$ is true, the forbidden combination, the clause is false.
\end{itemize}

    Notice that $x^{l_1}_{k_1}$ and $x^{l_2}_{k_2}$ can be in the same $X_k$ (when $k_1=k_2$), in which case the clauses from the first and second constraint will also be created.
    
    We see that by going through all these edges in $G$, we create at most $3(n-h)^2$ clauses, since for each edge there are maximum $3$ clauses. This, of course, holds because if there are four, we stop the procedure with the response that the homomorphism does not exist (since there is no edge in $H$ where we could map an edge $x^{l_1}_{k_1}x^{l_2}_{k_2}\in (V(G)-V(G[X]))$).
    
    \item Let $x^l_k \in (V(G)-V(G[X])), l\in [|X_k|], k\in [h+1]$ and $v_j\in V(G[X]), j\in [h]$ be adjacent vertices in $G$.
    Then we check two different mappings of $x^l_k\in V(G)$ to $H$ (since $v_j\in V(G[X])$ always maps to $v_j\in V(H)$):
    
\begin{align*}
    &x^l_k\to v_{k - 1} \\
    &x^l_k\to v_{k}.
\end{align*}

    For each of these cases, if there does not exist an edge in $H$, we create the following restrictive clause in our \twosat:
    
\begin{itemize}
    \item For the case where $x^l_k\to v_{k - 1}$, if there is no edge $v_{k - 1} v_j$ in $H$, then we create a clause $( s^l_k \lor s^l_k)$. We see that this way, if $s^l_k$ is false, the forbidden setting, our clause is false.
    \item For the case where $x^l_k\to v_{k}$, if there is no edge $v_k v_j$ in $H$, then we create a clause $(\neg s^l_k \lor \neg s^l_k)$. We see that this way, if $s^l_k$ is true, the forbidden setting, our clause is false.
\end{itemize}
    
    We see that by going through all these edges in $G$, we create at most $h^2$ clauses, since for each edge there is a maximum one clause (as if there are two then we stop the procedure with the response that the homomorphism does not exist for the same reason as above).
    
    Notice that the case of vertices $v_{j_1}, v_{j_2}\in V(G[X]), j_1, j_2\in [h]$, where there is an edge in $G[X]$, does not need to be treated due to the (identity) character of \ret{} (and therefore this edge also being in $H$).
    
    \item For all vertices $x^{l_1}_1\in X_1, l_1\in [|X_1|]$ and $x^{l_2}_{h+1}\in X_{h+1}, l_2\in [|X_{h+1}|]$, map these to the following vertices in $V(H)$:
    
\begin{align*}
    &x^{l_1}_1\to v_{1} \\
    &x^{l_2}_{h+1}\to v_{h}.
\end{align*}

    For each of these vertices, we create the following clause in our \twosat:
    
\begin{itemize}
    \item For each $x^{l_1}_1\in X_1$, create a clause $( s^{l_1}_1 \lor s^{l_1}_1)$. We see that this way, if $s^{l_1}_1$ is false, the forbidden setting, our clause is false.
    \item For each $x^{l_2}_{h+1}\in X_{h+1}$, create a clause $(\neg s^{l_2}_{h+1} \lor \neg s^{l_2}_{h+1})$. We see that this way, if $s^{l_2}_{h+1}$ is true, the forbidden setting, our clause is false.
\end{itemize}
    
    We see that by going through all these vertices in $X_1$ and $X_{h+1}$, we create the $|X_1| + |X_{h+1}|$ clauses.
    
\end{enumerate}

\paragraph{Equivalence of instances.} Suppose now that we have a solution of our \ret{} problem — an ordered homomorphism $f:G\to H$, where $f$ restricted to $H$ is an identity. As the restrictions applying to our ordered homomorphism $f$ were all reflected in the clauses we created in our \twosat, there cannot be any clause in our \twosat ~with negative value.

Suppose that we have a yes-instance of our \twosat ~constructed following the steps above.

We define $f : V(G) \to V(H)$ as follows. 

\begin{enumerate}
    \item if variable $s^l_k$ is false, then $f$ maps vertex $x^l_k \in (V(G)-V(G[X]))$ to $v_{k-1}\in H$,
    \item if variable $s^l_k$ is true, then $f$ maps vertex $x^l_k \in (V(G)-V(G[X]))$ to $v_{k}\in H$,
    \item Map all the vertices $v_i \in V(G[X])$ to $v_i\in H$
\end{enumerate}

Clearly, $f$ restricted to $G[X]$ is an identity.   

We see that for all vertices $x^l_k \in (V(G)-V(G[X]))$ that are mapped to $v_{k-1}\in H$ whenever the variable $s^l_k$ of our instance of the \twosat ~problem is false, and for all vertices $x^l_k \in (V(G)-V(G[X]))$ that are mapped to $v_{k}\in H$ whenever the variable $s^l_k$ from our instance is true, these respect the vertex ordering (from the first, second and fifth rules above) and preserve all edges (from the third and fourth rules above).

The last thing to address are the vertices $x^l_k \in (V(G)-V(G[X]))$ for which there is no variable $s^l_k$ in any of the clauses of our \twosat ~problem (if any). We notice that this can happen when, e.g., $x^l_k$ is not connected to any other vertex in $G$, $|X_k|=l=1$, and $k\notin \{1,h+1\}$. But this simply means that there is no restriction for these vertices, and they can be mapped to $v_{k-1}$ or $v_{k}\in H$.

As the maximum number of clauses in \twosat ~is the following, the claimed polynomial complexity of \ret{} follows.

$$
2 \sum_{k=1}^{h+1} \binom{|X_k|}{2} + 3(n-h)^2 + h^2 + |X_1| + |X_{h+1}|
$$

\end{proof}

\section{Hard Ordered Core Problems}

Let us start this section by a (standard) definition of \emph{hypergraph}, being a generalization of a graph in which edges can connect any number of vertices, not just two.

Formally, a hypergraph is an ordered pair $H=(V,E)$, where $V$ is a finite set of vertices and $E\subseteq P(V)\setminus \{\varnothing \}$ is a set of hyperedges, each hyperedge being a nonempty subset of $V$.

Then an \emph{ordered hypergraph} is a hypergraph with totally ordered vertex set.

Let us now define an \emph{ordered $k$-uniform hypergraph}, as a hypergraph with totally ordered vertex set, where each hyperedge is of size $k$.

We show \NP-completeness of \core ~for connected ordered $k$-uniform hypergraphs where $k\ge 3$.

\begin{theorem}
\label{thm:corehyper}
    \core ~for connected ordered $k$-uniform hypergraphs is \NP-complete for $k\ge 3$.
\end{theorem}

\begin{proof}
    We prove \NP-hardness by a reduction from ONE-IN-THREE SAT. We require that the ONE-IN-THREE SAT formula is \emph{connected}, meaning it cannot be partitioned into two or more disjoint sets of clauses that are variable-disjoint. We note that this means that we cannot simplify the problem by breaking it apart into more smaller problems (or formulas) in this way. This promise variant of ONE-IN-THREE SAT is, of course, still \NP-complete by ~\cite{SchaeferSATComplexity1978}.
    

    We define an ordered $3$-uniform hypergraph $G$, such that $G$ is not a core (or the there exists non-surjective endomorphism on $G$), if and only if there exists a satisfying assignment of the ONE-IN-THREE SAT formula. 
    
    In the following, we may call a hyperedge of size three simple a \emph{triple}.

    \paragraph{Definition of $G$.}
    Let us define a variable gadget as an ordered hypergraph on four vertices $1,2,3,4$ and a triple $\{1,2,4\}$. We then add a variable gadget in G per each variable in the ONE-IN-THREE SAT formula. Let $v$ be a number of variables in the formula. The graph $G$ will therefore have $4v$ vertices.

    Then for each clause with variables $v_i,v_j,v_k$ in ONE-IN-THREE SAT formula, we create a triple connecting third vertices of variables gadgets corresponding to $v_i,v_j,v_k$, calling this triple \emph{dynamic}. We also create three additional triples connecting second, fourth, and fourth vertices of the variables gadgets corresponding to $v_i,v_j,v_k$, respectively, in the similar manner we create a second triple connecting fourth, second, and fourth vertices of these variable gadgets, and we also create a third triple connecting fourth, fourth, and second vertices of these three variable gadgets, respectively, calling these triples \emph{static}. 
    
    This concludes the definition of $G$. See Figure ~\ref{pic:CoreHyper} for variable gadgets, dynamic triples (with thick edges), and one (out of three) static triple for the blue clause (we do not include more static triples in order to keep the picture readable). Colors illustrate which clause belongs to which triple.

    \begin{figure}[p]\centering
    \includegraphics[width=\textwidth,height=\textheight,keepaspectratio]{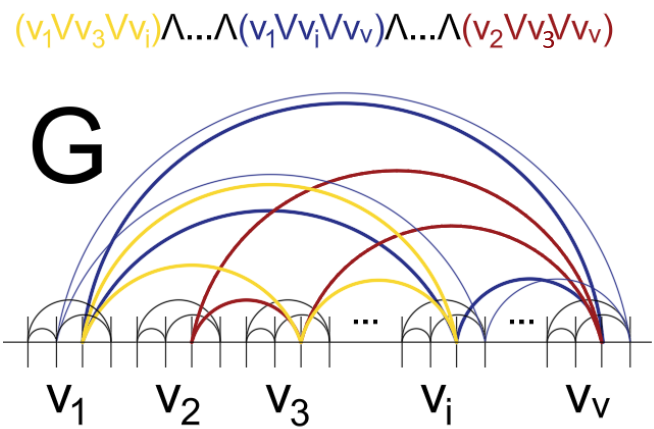}
    \caption{Ordered Hypergraph $G$ from Theorem ~\ref{thm:corehyper}.}
    \label{pic:CoreHyper}
    \end{figure}

    We then note that dynamic triples are the only triples that can be mapped to a subgraph of $G$. We also notice that they can map only to three choices of static triples per each dynamic triple. We will then interpret mapping a dynamic triple to a second vertex of $v_i$ variable gadget, fourth vertex of $v_k$ variable gadget, and fourth vertex of $v_l$ variable gadget as setting a variable $v_i$ to true and variables $v_k, v_l$ to false, and other configurations analogously.

    We then notice that since the ONE-IN-THREE SAT formula is connected, our hypergraph $G$ must be connected, and by mapping one dynamic triple to one of the static triples, we will need to map all the dynamic triples to static triples, where each dynamic triple has only three options of static triples to map to. For contradiction, suppose this is not the case. Then, there is at least one dynamic triple that does not map to one of the static triples. But then either this dynamic triple is not connected to other dynamic triples, in which case $G$ would be disconnected, or all other dynamic triples would not map to static triples. In which case we would not get an ordered non-surjective retraction $G\to G$, as mapping the dynamic triples to the static ones was the only option, a contradiction.

    Also, one easily notes that shared vertices of different dynamic triples must map to the same vertex. Therefore, the same variable in different formulas will map to the same value.

    \paragraph{Equivalence of instances.} Suppose that we have a YES-instance of our ONE-IN-THREE SAT formula. Then for each dynamic triple in its corresponding clause in $G$ we map the first vertex of the dynamic triple to the second or fourth vertex of a variable $v_i$, depending on whether the assignment of the variable $v_i$ in the formula is true or false, respectively, and analogously for the remaining two variables in the clause. In this way, we clearly get the non-surjective ordered retraction $G\to G$.

    \medskip

    Now suppose that we have a non-surjective ordered retraction $G\to G$. Then for each $v_i$ variable gadget that mapped the third vertex to the second vertex of the variable gadget $v_i$, we assign true value to the variable $v_i$ in the ONE-IN-THREE SAT formula. Similarly, for each variable gadget that mapped the third vertex of the variable gadget $v_j$ to the fourth vertex of the variable gadget $v_j$ we assign false value to the variable $v_j$ in the ONE-IN-THREE SAT formula. Since $G$ is connected, the non-surjective ordered retraction $G\to G$ must map all the variables to true or false.

    We then note that by adding $n$ vertices in between the first and second vertex in each variable gadget in $G$ and connecting every static and dynamic triple to these vertices as well as connecting every variable gadget triple to these vertices, we can create $n+3$-uniform hypergraph $G$. It is easy to see that this $n+3$-uniform hypergraph $G$ has the same properties with respect to the reduction from the ONE-IN-THREE SAT problem.

    \core for connected ordered $k$-uniform hypergraphs is in \NP as given its solution (non-surjective homomorphism of the given ordered graph to itself), we can verify it in polynomial time.

    This completes the proof.
\end{proof}

The following result then naturally follows.

\begin{corollary}
\label{cor:rethypergraphsNPC}
    \ret ~for $k$-uniform hypergraph is \NP-complete for $k\ge 3$.
\end{corollary}

\begin{proof}
    We see that in the proof of the Theorem ~\ref{thm:corehyper} the ONE-IN-THREE SAT formula has a satisfying assignment if and only if the hypergraph $G$ is not a core and maps to its subgraph $G'$. The vertices (and triples) of the hypergraph $G'$ are always precisely defined within the hypergraph $G$ (as $G$ without all the dynamic edges). The result therefore follows.
\end{proof}


Let us now define a computational problem, which we will call \slice{}. It will be naturally related to the \core problem, where the retract of $G$ has a particular number of vertices and edges. Let $n=|V(G)|, m=|E(G)|$.

\begin{prb}
\label{Prb:slicegh}
\problemStatement{\slice{}}
  {Ordered graph $G$ and functions $g(|V(G)|):\mathbb{N}\to \mathbb{N}$ and $h(|V(G)|,|E(G)|):\mathbb{N}^2\to \mathbb{N}$, where $ g(|V(G)|)< |V(G)|; h(|V(G)|,|E(G)|)< |E(G)|$.}
  {Does there exist an ordered graph $H\subset G; |V(H)|=g(|V(G)|); |E(H)|=h(|V(G)|,|E(G)|)$, such that there exists an ordered homomorphism $f:G\to H$?}
\end{prb}

The following result proves the \NP-completeness of the \slice, or again, in other words, the \core problem under the assumption that the retract is of a particular size.

\begin{theorem}
\label{thm:Subgraph}
\slice{} is \NP-complete.
\end{theorem}

\begin{proof}
In order to prove that our problem is \NP-hard, we again use a reduction from ONE-IN-THREE SAT. We will require that each variable is in at least three clauses. This will clearly preserve the \NP-complete nature of the ONE-IN-THREE SAT problem (see, e.g., ~\cite{SchaeferSATComplexity1978}).


\paragraph{Definition of $G$.}
To construct an ordered graph $G$, let us first construct a variable gadget as an ordered graph with four vertices $1,2,3,4$ (in this order) and an edge set 
$$
E(G)=\{\{1,3\},\{2,4\},\{3,4\}\}.
$$
We shall call these edges within the variable gadget the \emph{variable edges}.

For clause gadget, glue three of these variable gadgets together, where by gluing we will mean that the last vertex of the previous (in the order of $G$) variable gadget is identified with the first vertex of the following variable gadget. Each clause gadget will therefore have three variable gadgets, and we glue the $c$ clause gadgets ($c$ being a number of clauses in the formula) so that again the last vertex of the previous clause gadget is identified with the first vertex of the following clause gadget. The ordered graph $G$ will therefore have $n=c(3\times3)+1$ vertices. 

Next, we will connect the second vertex of each variable gadget with second and third vertices in other variable gadgets within the same clause gadget, and the third vertex of each variable gadget with second and third vertices in other variable gadgets within the same clause gadget. We shall call these edges \emph{clause edges}.

For each variable $x$ in our formula, we identify its variable gadgets in $G$ (from the definition of our problem, there need to be at least three of them). Let us say that the variable $x$ occurs in the $o_x$ clauses. We will then connect the second vertex of the $i$-th occurrence of the variable gadget $x$ with the second vertex of the $i+1$-th occurrence of the variable gadget $x$, where $i=1,\ldots,o_x-1$ and we also connect the second vertices of the last and first occurrence of the variable gadget $x$. We will then also connect the third vertex of the $i$-th occurrence of the variable gadget $x$ with the third vertex of the $i+1$-th occurrence of the variable gadget $x$ , where $i=1,\ldots,o_x-1$ and we also connect the third vertices of the last and first occurrences of the variable gadget $x$. We shall call these edges \emph{external edges}.

Notice that a variable does not necessarily need to be in every clause; therefore, we connect variable gadgets only to the same variable gadgets in other clause gadgets. As each variable occurs in three or more clauses, these edges will form (two per variable $x$) cycles of length $o_x$. Let us denote $m=|E(G)|$.

\paragraph{Definition of $g(n)$ and $h(n,m)$.} 
Let us now define a set of ordered graphs $\mathcal{H}$. We first define the functions $g$ and $h$ as follows.

\begin{align*}
&g(n)=n - \frac{n-1}{9},\\
&h(n,m)=m - 6\times\frac{n-1}{9}.
\end{align*}

We see that with $n=9c+1$, these functions are in $\mathbb{N}$ and that $g(n)$ is equal to the difference between $n$ and the number of clauses $c=\frac{n-1}{9}$ (and $h(n,m)=m-6c$). Hence, since $c>0$, we get $g(n)< n$ and $h(n,m)< m$.

We then define a class of ordered graphs $\mathcal{H}$ that will contain all the proper ordered subgraphs of $G$ on the $g(n)$ vertices and the $h(n,m)$ edges.

\paragraph{Equivalence of Instances.} 
We now claim that there exists a solution to ONE-IN-THREE SAT with each variable in at least three clauses if and only if there is an ordered homomorphism $f$ from $G$ to one of the ordered subgraphs in $\mathcal{H}$. But before we get to the proof of this equivalence, we notice a couple of observations about the nature of $G$ and the ordered graphs in $\mathcal{H}$.

The first observation is that each variable gadget in $G$ can be mapped to three possible ordered graphs. One is identity, the second is $P_2$, and the third is $K_3$. We notice that only the first two will be an ordered homomorphism image that is a subgraph of $G$ (since $K_3$ would add an edge from the first to the third vertex).

Therefore, the only way to map $G$ to its proper ordered subgraph with the restrictions of our problem is to map some of its variable gadgets to $P_2$. By mapping a variable gadget to $P_2$ we decrease the number of vertices by one and also the number of variable edges by one. As a result, since we are decreasing the number of vertices by $c$ (from the definition of $g(n)=n-c$), we will also decrease the number of variable edges by exactly $c$.

Notice that the ordered graphs in $\mathcal{H}$ are not on the $\chi(G)$ vertices, because the $\chi(G)=c(3\times3)+1-3c=5c-1$, since each variable gadget on four vertices can be mapped to an ordered graph on three vertices, and there are $3c$ variable gadgets in $G$. In our problem, we are mapping only the $c$ variable gadgets to ordered graphs $P_2$, instead of the $3c$ variable gadgets.

Now we need to ensure that if we map a variable gadget in one clause to $P_2$, we will map that variable gadget to $P_2$ also in all other clauses, where this variable appears. But this is ensured by the external edges connecting the same variables gadgets in different clauses.

Assume otherwise, or that a variable gadget in $G$ is mapped to $P_2$ in a clause gadget and the same variable gadget in the other clause gadget (connected by two external edges) is mapped by identity. But then the variable gadget mapped to $P_2$ will have an additional edge coming from the second vertex, and the ordered homomorphic image will not be a subgraph of $G$. Now, if both variables are mapped by identity, the resulting graph is a subgraph, and if they are both mapped to $P_2$ the resulting graph is also a subgraph. Therefore, the same variables always need to map the same way.

Here we note an important observation that by having each variable in at least three clauses, by decreasing the number of vertices by $c$ we will decrease the number of external edges also by $c$, as we are mapping $c$ variable gadgets to $P_2$ and for each variable $x$ occurring in $o_x, o_x>2$ clauses, we have $2o_x$ external edges formed by two cycles of size $o_x$ which are mapped to each other. Thus, again, by decreasing the number of vertices by $c$ we will also decrease the number of external edges by $c$.

Another point we will need to ensure is that only one variable gadget per each clause gadget in $G$ will be mapped to $P_2$ (from the definition of the ONE-IN-THREE SAT). But we notice that by mapping one variable gadget to $P_2$ in a clause, we decrease the number of clause edges by four. By mapping two variable gadgets in the same clause to $P_2$, the second variable gadget mapped to $P_2$ will decrease the number of clause edges by three. Lastly, by mapping three variable gadgets in the same clause to $P_2$, the last variable gadgets mapped to $P_2$ will decrease the number of clause edges by two.

We notice that since the number of variable edges has already decreased by $c$ and external edges by $c$ too, we need to decrease the number of remaining clause edges by $4c$ (since $h(n,m)=m-6c$). Therefore, since we are decreasing the number of vertices by $c$, this is only possible if by mapping a variable gadget in $G$ to $P_2$, we decrease the number of clause edges by four. But this is only possible if all the $c$ variable gadgets in $G$ mapped to $P_2$ are in different clauses.

Now we should be ready to show our equivalence. Let us assume that there is a true ONE-IN-THREE SAT formula with each variable in at least three clauses. Constructing $G$ as described above and for each true variable mapping its variable gadgets to $P_2$, we obtain our ordered retraction from $G$ to $H\in \mathcal{H}$ as defined in our \slice{} problem. 

Conversely, having an non-surjective ordered retraction $f$ from $G$ to $H\in \mathcal{H}$ as defined above, we simply choose a true value for every variable gadget mapped by $f$ to $P_2$ and a false value for every variable gadget mapped by identity.

In order to prove that this ordered retraction $f:G\to H, H\in \mathcal{H}$ must satisfy $f(v)=v$ for all $v\in V(H)$ we first see that because each variable gadget contains $P_2$ and these $P_2$ paths are connected by a common vertex, we cannot map a vertex of one variable gadget to another variable gadget. The only variable gadgets in $G$ that are not mapped by identity are the variable gadgets that are mapped to $P_2$. This can happen only by mapping a third vertex to a second vertex of the variable gadget, or vice versa.

We see that mapping a third vertex to a second would not be an identity on $H$ as this would connect the first and second vertex of the variable gadget and the first and second vertex are not connected in the variable gadget in $G$. We notice that by mapping the second vertex of the variable gadget in $G$ to the third, we get the identity on $H$. But since the third vertex can be mapped to the second vertex within the variable gadget in $G$, if and only if the second vertex can also be mapped to the third vertex within the same variable gadget and these different operations result in the same ordered graph $H$, we can assume that the latter operation is always performed. In this way, an ordered homomorphism $f:G\to H, H\in \mathcal{H}$ must satisfy $f(v)=v$ for all $v\in V(H)$, and we get an ordered retraction $f$.

We also see that by construction, all the ordered graphs in $\mathcal{H}$ must be proper ordered subgraphs of $G$.

Our problem is in \NP, since given its solution, we can verify it in polynomial time. In particular, we can verify whether $H$ has $g$ vertices and $h$ edges, and by Theorem ~\ref{thm:RetInPP}, whether $G\to G(X)=H$ is an ordered retraction.

This completes the proof.
\end{proof}

This result will now help us derive the complexity of a more general problem.

First, let us define a \emph{doubletuple}, as a set of two sets of $l$ and $k$ integer functions, respectively, with all functions dependent on $n\in \NN$.
$$
\{t,u\}=\{\{t_1(n),t_2(n),\ldots,t_l(n)\},\{u_1(n),u_2(n),\ldots,u_k(n)\}\}
$$ 

Putting $n\in \NN, t\in\{t,u\}$, we will then denote by $\{n-t\}$, a set of integer functions $\{n-t_1(n),n-t_2(n),\ldots,n-t_l(n)\}$. By a set of graphs on $\{n-t\}$ vertices and $\{m-u\}; n,m \in \NN; t,u \in\{t,u\}$ edges we will then mean the set of graphs on a product of $\{n-t\}$ vertices and $\{m-u\}$ edges, product meant in its usual sense $\{n-t\}\times \{m-u\}$. I.e., as ordered graphs $G$, such that, 

$|V(G)|=n-t_1(n)$ and $|E(G)|=m-u_1(n)$, or

$|V(G)|=n-t_1(n)$ and $|E(G)|=m-u_2(n)$, or

$\ldots$

$|V(G)|=n-t_1(n)$ and $|E(G)|=m-u_k(n)$, or

$|V(G)|=n-t_2(n)$ and $|E(G)|=m-u_1(n)$, or

$\ldots$

$|V(G)|=n-t_2(n)$ and $|E(G)|=m-u_k(n)$, or

$\ldots$

$|V(G)|=n-t_l(n)$ and $|E(G)|=m-u_1(n)$, or

$\ldots$

$|V(G)|=n-t_l(n)$ and $|E(G)|=m-u_k(n)$.

Let us now define the following computational problem, which we will call the \sub{}. We will denote $n=|V(G)|$ and $m=|E(G)|$.

\begin{prb}
\label{Prb:sub}
\end{prb}
\problemStatement{\sub{}}
  {Ordered graph $G$ and a doubletuple $\{t,u\}$, where $0<t_1(n)<n,0<t_2(n)<n,\ldots,0<t_l(n)<n$ and $0<u_1(n)<m,0<u_2(n)<m,\ldots,0<u_k(n)<m$.}
  {Let $\mathcal{H}$ be a class of proper ordered subgraphs of $G$ on $\{n-t\}$ vertices and $\{m-u\}$ edges. Does there exist $H\in \mathcal{H}$ such that there exists an ordered retraction $f:G\to H$?}


\begin{corollary}
\sub{} is \NP-complete.
\end{corollary}

\begin{proof}
It is easy to see that \sub{} is a generalization of \slice{}, since we can construct a reduction from \slice{} to \sub{} simply by restricting the doubletuples $\{t, u\}$ so that $t$ contains only one function, as well as $u$ contains only one function (both dependent on $n$). We see that we can do this since $\{n-t\}$ is the set of integer functions dependent on $n$ and $\{m-u\}$ is the set of integer functions dependent on $n$ and $m$, which was the case for two functions $g$ and $h$, respectively, in \slice{}. This would give us an identical problem: if we have the solution to our \slice{}, we will have the solution to our \sub{}, with a restricted size of doubletuples containing exactly two functions (with the problem ~\ref{Prb:sub} criteria), and vice versa.

We also see that \sub{} is in \NP as given its solution, we can verify it in polynomial time, using the same arguments as in the proof of Theorem ~\ref{thm:Subgraph}.
\end{proof}


We see that \core is, of course, "easier" than \sub, since \sub ~is a generalization of \core.
Searching for a reduction from \slice to \core (to prove the hardness of \core) is also not that straightforward, since while \slice{} is looking for a particular size of the subgraphs, \core{} seeks any of them.

Consequently, determining the complexity of \core for ordered graphs (or $2$-uniform hypergraphs) remains unresolved using the above approaches (which is not to say that the above results (in our view) are not interesting). Let us therefore try to choose a different way.

Let us with the definition of the following computational problem, which we will denote \corechi. 

\begin{prb}
\label{Prb:CorechiOne}
\problemStatement{\corechi}
  {Ordered graph $G$ with $\chi^<(G)=k$.}
  {We want to decide between two cases:
  \begin{enumerate}
      \item The core of $G$ has $k$ vertices.
      \item $G$ is a core.
\end{enumerate}}
\end{prb}

We will prove the \NP-hardness of \corechi and \wone-hardness of \corechi parameterized by $k$.
This result implies, in particular, \NP-hardness of the \core and the \wone-hardness of \core parameterized by the number of vertices $|V(H)|$ of the image ordered subgraph $H$ of $G$.

Let us start with definitions of the following class of ordered graphs. The ordered graph $G$ is \emph{edge-collapsible} if and only if the only ordered subgraph $H$ of $G$, for which there exists an ordered homomorphism $G\to H$, is $P_2=H$.

\label{def:CollapsibleM4}

Next, let $M^C_4$ be an ordered graph on eight vertices $v_1,\ldots,v_8\in V(M^C_4)$ and the following edges,
$
E(M^C_4)=\{\{v_1,v_6\},\{v_2,v_8\},\{v_3,v_5\},\{v_4,v_7\}\}.
$

Before we proceed further, we note that, for any ordered matching, each edge can be uniquely referred to by one of the vertices to which it is incident. For example, we can refer to an edge $\{v_3,v_5\}\in E(M^C_4)$ by the vertex $v_3$. We will use the following notation.
Let $M$ be an ordered matching; then we denote an edge $(v_i,v_j)\in E(M);v_i,v_j\in V(M)$ as $v_i^e$ or $v_j^e$.
We can then refer to a set of edges $\{v_{i_1},v_{i_2}\},\ldots,\{v_{i_3},v_{i_4}\}$ in ordered matching as, e.g., $v^e_{i_1},\ldots,v^e_{i_3}$. We also note that matchings play a crucial role in the area of ordered homomorphisms (see ~\cite{nescer2023duality}, ~\cite{certik_matching_2025}).

We start with the following lemma. 

\begin{lemma}
\label{lem:MC4Collapsible}
    $M^C_4$ is edge-collapsible.
\end{lemma}

\begin{proof}
    For contradiction, assume that $M^C_4$ is not edge-collapsible.

    There are three ways in which $M^C_4$ can be not edge-collapsible.

    \begin{enumerate}
        \item If $M^C_4$ is an ordered core. However, there is an ordered homomorphism from $M^C_4$ to its ordered subgraph $P_2$.
        \item If there exists an ordered homomorphism that maps $M^C_4$ to an ordered subgraph $M^C_4[X]$ of $M^C_4$, where $|E(M^C_4[X])|=3$. This is only possible if there is a pair of edges in $M^C_4$ that maps to one of these edges in $M^C_4$. Since ordered matching $M^C_4$ is fixed and finite, there is a finite number of these pairs, and we can investigate them one-by-one. The following are all the options.

        $$
        \{\{v^e_1,v^e_2\},\{v^e_2,v^e_3\},\{v^e_3,v^e_4\},\{v^e_1,v^e_3\},\{v^e_2,v^e_4\},\{v^e_1,v^e_4\}\}
        $$
        
        However, none of these pairs of edges in $M^C_4$ can be mapped to the one of the edges within the pair, since there is always at least one other edge in $M^C_4$ preventing that (in the sense that also this other edge would need to map together with them if this was to be an ordered homomorphism to an ordered subgraph of $M^C_4$).

        \item If there exists an ordered homomorphism that maps $M^C_4$ to an ordered subgraph $M^C_4[X]$ of $M^C_4$, where $|E(M^C_4[X])|=2$. This is only possible if there are three edges in $M^C_4$ that map to one of these edges in $M^C_4$. The following are all the options.

        $$
        \{\{v^e_1,v^e_2,v^e_3\},\{v^e_1,v^e_2,v^e_4\},\{v^e_1,v^e_3,v^e_4\},\{v^e_2,v^e_3,v^e_4\}\}
        $$
        
        However, none of these triples of edges in $M^C_4$ can be mapped to the one of the edges within the triple, since there is always at least one other edge in $M^C_4$ preventing that.
    \end{enumerate}

Therefore, we arrive to a contradiction and $M^C_4$ is edge-collapsible, since we also showed that it can be mapped to $P_2$.
\end{proof}

Notice that there does not exist an ordered matching $M; M\to P_2; |E(M)|=3$
that is edge-collapsible (since all such matchings can be mapped to their ordered subgraph with \emph{two} edges).

It is also easy to see that all ordered matchings 
$M; M\to P_2 ; |E(M)|=2$
are edge-collapsible (since $M\to P_2$ ensures that $M$ can be mapped to an edge).

We now define in an iterative way the following set of ordered matchings $M^C_i, |E(M^C_i)|=i, i>3$.
Let $M^C_4$ be an ordered matching defined above. Then an ordered matching $M^C_i,i>4$ is defined in the following way.

Take an ordered matching $M^C_{i-1}$ with vertices $v_1,\ldots,v_{2(i-1)}$. If $i$ is odd, add the vertex $w$ between the vertices $v_{i-1}$ and $v_i$ in $M^C_{i-1}$, add the vertex $v$ in front of the first vertex $v_1\in M^C_{i-1}$, and connect $v$ and $w$ by an edge $e$. If $i$ is even, add the vertex $w$ after the vertex $v_{2(i-1)}\in M^C_{i-1}$, add the vertex $v$ in front of the first vertex $v_1\in M^C_{i-1}$, and connect $v$ and $w$ by an edge $e$.
$M^C_{i},i>4$ is then a (disjoint) union of $M^C_{i-1}$ and $e$. It is clear that each $M^C_{i},i>4$ is well defined.

Let us now show that all ordered matchings $M^C_i, |E(M^C_i)|=i, i>3$ are edge-collapsible. 

\begin{lemma}
\label{lem:MCiCollapsible}
    For every $i>3, i\in \NN$, $M^C_i$ is edge-collapsible.
\end{lemma}

\begin{proof}
    Let us prove the statement by induction on a number of edges $i$. 
    
    The base case for $i=4$ holds by Lemma ~\ref{lem:MC4Collapsible}.

    Assume that the statement holds for $i$. By definition, the new edge $e$ is added in a way that $\chi^<(M^C_{i+1})=2$, since $v$ is the leftmost vertex and $w$ is between $v_i\in M^C_{i}$ and $v_{i+1}\in M^C_{i}$ or $w$ is placed after the rightmost vertex $v_{2i}\in M^C_{i}$. Therefore, $M^C_{i+1}$ can map to $P_2\subset M^C_{i+1}$.

    The remaining part is to show that $M^C_{i+1}$ cannot map to any other ordered subgraph than $P_2$. We prove this by contradiction and assume that $M^C_{i+1}$ is not edge-collapsible.
    
    We observe that by adding the edge $e$ to $M^C_{i}$, we cannot make the ordered subgraph $M^C_{i}$ of $M^C_{i+1}$ \emph{not} edge-collapsible. Therefore, the only way $M^C_{i+1}$ could be not edge-collapsible is if the edge $e$ mapped to some other edge in $M^C_{i}$, or the edge in $M^C_{i}$ mapped to $e$. We notice that the only edge where $e$ could possibly map (or the edge that could map to $e$) is the edge $v^e_1\in M^C_{i}$, since the vertex $v$ has only the vertex $v_1\in M^C_{i}$ next to it.

    Let us denote a vertex $v_{1_j}\in M^C_{i}$, the vertex adjacent to $v_1\in M^C_{i}$. However, since by construction, the vertices $w$ and $v_{1_j}$ have at least one vertex $v_{w1_j}\in M^C_{i}$ between them, at least the edges $v^e_{w1_j}$, $e$, and $v^e_{1_j}$ must be mapped together to one of these three edges if this was ordered homomorphism to an ordered subgraph of $M^C_{i+1}$. But by the induction hypothesis, since $M^C_{i}$ is edge-collapsible, mapping edges $v^e_{w1_j}$ and $v^e_{1_j}$ to one edge is not possible without mapping all other edges in $M^C_{i}$ to just one edge, and we arrive to a contradiction.

    We only add that the first step of the iteration when $i=5$, $M^C_{5}$ is edge-callapsible, using the same arguments as above and Lemma ~\ref{lem:MC4Collapsible}.
\end{proof}

Now we should be ready to prove the main statement.

\begin{theorem}
\label{thm:corechiNPC}
    The following hardness results hold. The \corechi problem
    \begin{enumerate}
        \item      is \NP-hard,
    \item  is \wone-hard, when parameterized by $k$,
    \item cannot be solved in time $2^{o(|V(G)|)}$ nor in time $n^{o(k)}$, unless the ETH fails.
    \end{enumerate}
\end{theorem}

\begin{proof}
To prove the statement, we proceed with the reduction from the \textsc{Multicolored Clique} problem. The problem is defined as follows. Let $F$ be the instance graph whose vertex set is partitioned into $k$ subsets $V_1, V_2,\ldots,V_k$, each $V_i,i \in [k]$ containing an independent set of vertices. We ask if $F$ has a clique, intersecting every set $V_i,i\in [k]$.

By copying some vertices if necessary, without loss of generality, we may assume that for each $i \in [k]$ we have $|V_i|=l$, and that $l>3$. The problem \textsc{Multicolored Clique} is \NP-complete and \wone-hard when parameterized by $k$; furthermore, it cannot be solved in time $(k\ell)^{o(k)}$, unless the ETH fails (see, e.g., ~\cite{gareyjohnson_compintr_1982}, ~\cite{1055552815661} and ~\cite{chen2006strong-b6f}).

We start by fixing an arbitrary total order for each $V_i = \{x^i_1,\ldots,x^i_{l}\}, i \in [k]$ in $F$.

Let us now define an ordered graph $G$, based on the input instance graph $F$ of the \textsc{Multicolored Clique} problem.

\begin{defn}[Ordered Graph $G$ from \textsc{Multicolored Clique} to \corechi reduction]
\label{def:GCoreNPC}
\hfill
\end{defn}

For each $i \in [k]$, let $C_i$ be a set of $l$ vertices denoted by
$
C_i=\{v^i_1,\ldots,v^i_l\},
$
which also determines the ordering of the vertices in $C_i$. The sets $C_i$ are mutually disjoint.

Next, let $B_i,i\in[k]$ be a set of $k-1$ vertices denoted as follows.

$$
B_i=\{w^i_1,w^i_2,\ldots,w^i_{i-1},w^i_{i+1},\ldots,w^i_{k}\}
$$

This also determines the order of the vertices in $B_i$, and all sets $B_i,i\in[k]$ are mutually disjoint.

We also define $D_i,i\in[k]$, as a set of $l+k-1$ vertices, in the following way.

$$
D_i=\{d^i_1,\ldots,d^i_l,d^i_{l+1},\ldots,d^i_{l+k-1}\}
$$

This also determines the order of the vertices in $D_i$, and all the sets $D_i,i\in[k]$ are again mutually disjoint.

Now we define a set of vertices $A_i,i\in[k]$ consisting of the sets $C_i, B_i$, where the vertices in $C_i$ are followed by the vertices in $B_i$, as follows,
$
A_i=C_i,B_i.
$
The sets $A_i,i\in[k]$ are mutually disjoint.

We order these sets of vertices in the following way, adding a vertex in the beginning, at the end, and in between sets $D_i,i\in[k]$ and $A_i,i\in[k]$ in the following way.

$$
p_1,D_1,p_2,D_2,\ldots,p_k,D_k,p_{k+1},A_1,p_{k+2},A_2,\ldots,p_{2k},A_k,p_{2k+1}
$$

For fixed $i\in [k]$, we then connect the vertex $p_i$ with all vertices in $D_i$, and we connect the vertex $p_{i+1}$ with all vertices in $D_i$. Also, for fixed $i\in \{k+1,k+2,\ldots,2k\}$, we connect the vertex $p_i$ with all the vertices in $C_i$, and connect the vertex $p_{i+1}$ with all the vertices in $C_i$. We will call both sets of these edges \emph{path edges}.

Then, for all $i\in [k]$ and all $j\in[k], j\not =i$, we connect the vertices $w^j_i$ and $w^i_j$. We will call these edges \emph{complete edges}.

We also add another set of edges in the following way. If $x^{i'_1}_{j'_1}$ and $x^{i'_2}_{j'_2}; i'_1,i'_2\in[k];j'_1,j'_2\in [l]$ are adjacent in $F$ then we add an edge between $v^{i'_1}_{j'_1}$ and $v^{i'_2}_{j'_2}$ in $G$. We will call these edges \emph{original edges}.

We define the last set of edges by forming an ordered matching $M^C_{l+k-1}$ defined above on the vertices $D_i$ and $A_i$. We see that $|V(M^C_{l+k-1})|=|V(D_i)|+|V(A_i)|$, and since $l>3$, these edges will be well defined. We will call them \emph{collapsible edges}.

This completes the definition of $G$. Note that $G$ has $2k+1+2k(l+k-1)$ vertices, and $\chi^<(G)=2k+1+2k=4k+1$.

\emph{Equivalence of instances.}
We claim that there exists a clique in $F$, intersecting every set $V_i,i\in [k]$, if and only if there exists an ordered homomorphism $f$ from $G$ to an ordered subgraph $H$ of $G$.

In the remainder of the proof, we may keep the same notation of vertices in $H$ as we do in the ordered graph $G$ without risk of confusion, since $H$ is an ordered subgraph of $G$.

Suppose that $F$ has a yes-instance of \textsc{Multicolored Clique}, that is,
for each $i \in [k]$ there is $j_i$, such that $C_F = \bigcup_{i \in [k]} \{ x^i_{j_i}\}$ is a clique in $F$.

We define $f: V(G) \to V(H)$, where $H\subset G$, in the following way. Fix $i \in [k]$. We map the vertices in $A_i\in V(G)$ to $v^i_{j_i}\in V(G)$. This is feasible since the vertices in $A_i$ form an interval and $v^i_{j_i}$ is a vertex in $ A_i\in V(G)$. 

We also map the vertices in $D_i\in V(G)$ to the vertex $d^i_j$, which is adjacent by a collapsible edge with $v^i_{j_i}$. This is also feasible since the vertices in $D_i$ form an interval and $d^i_j$ is a vertex in $D_i\in V(G)$.

Then we map $p_i\in V(G)$ to $p_i\in V(G[H]), i=1,\ldots,2k+1$ by identity.

Let us show that $f$ is an ordered homomorphism. Clearly, by the construction of $G$, $f$ respects the ordering of vertices.

Let us show that it preserves the edges.
With fixed $i \in \{k+1,\ldots,2k+1\}$, we see that the path edges connecting $p_i$ with the vertices in $A_i$ in $G$ will map to the edge connecting the vertices $p_i$ and $v^i_{j_i}$ in $H$, and the path edges connecting $p_{i+1}$ with the vertices in $A_i$ in $G$ will map to the edge connecting $p_{i+1}$ and $v^i_{j_i}$ in $H$.

With fixed $i \in [k]$, the same will hold for the path edges connecting $p_i$ with the vertices in $D_i$ in $G$. They will map to the edge connecting the vertices $p_i$ and $d^i_j$ in $H$, and the paths that connect $p_{i+1}$ with the vertices in $D_i$ in $G$ will map to the edge connecting $p_{i+1}$ and $v^i_j$ in $H$.

For complete edges, we see that since there is a complete edge between every pair $B_{i'_1}$ and $B_{i'_2}; i'_i,i'_2\in[k]; i'_i\not =i'_2$ in $G$, if the vertices of each $B_i, i\in [k]$ map to a vertex $v^i_{j_i}$ in $H$, then the complete edges in $H$ will form a complete ordered graph $K_k$. To see this, we can assume for contradiction that there is no complete edge between vertices $v^i_{j_i}$ and $v^j_{j_j}$ in $H$. But there is a complete edge connecting $w^i_{j}$ and $w^j_{i}$ in $G$ that mapped to an edge connecting $v^i_{j_i}$ and $v^j_{j_j}$ in $H$, a contradiction. 

Now, since $C_F = \bigcup_{i \in [k]} \{ x^i_{j_i}\}$ is a clique in $F$, $C_G = \bigcup_{i \in [k]} \{ v^i_{j_i}\}$ is a clique in $G$, therefore, the mapping $f$ will preserve the complete edges.

$G[C_G]$ forming a complete ordered graph $K_k$ is also a reason why the original edges will be preserved, since they will all map to a complete ordered graph $G[C_G]$.

The last set of edges which preservation we will need to check are collapsible edges. But we showed in Lemma ~\ref{lem:MCiCollapsible}, that $M^C_{l+k-1}$ can map to a single edge, so since $\{d^i_j,v^j_{j_j}\}$ is an edge in $H$ and $l>3$, these edges are also preserved by $f$.

We can see that $H$ is an ordered subgraph of $G$ through the following reasoning. The complete ordered graph $G[C_G]\cong K_k$ must exist by the existence of the clique $C_F$ in $F$. Also, the directed ordered path $P_{4k+1}$ connecting all the vertices in $H$ is also the following ordered subgraph in $G$,
$$
p_1,d^1_{j_{v^1}},p_2,d^2_{j_{v^2}},\ldots,p_k,d^k_{j_{v^k}},p_{k+1},v^1_{j_1},\ldots,p_{2k},v^k_{j_k},p_{2k+1}.
$$
As we reasoned above, for each $i\in[k]$, the collapsing edges connecting $d^i_{j_{v^i}}$ and $v^i_{j_i}$ in $H$, are also present in $G$.
There are no other vertices or edges in $H$, and $H$ has $2k+1+2k=4k+1$ vertices. Since $G$ contains the above directed path $P_{4k+1}$, then $\chi^<(G)\ge 4k+1$. So, since $|V(H)|=\chi^<(H)=4k+1$, we get $|V(H)|=\chi^<(G)$.


Now consider an ordered homomorphism $f: V(G) \to V(H), H\subset G$.


Let us first show the following two lemmas, which will help us determine to which vertices in a proper ordered subgraph $H$ of $G$, must vertices in $G$ map.


\begin{lemma}
\label{lem:pitopiNPC}
    Let $G$ be an ordered graph defined in Definition ~\ref{def:GCoreNPC}, and $f:V(G)\to V(H), H\subset G$ be an ordered homomorphism. Then for each $i\in [2k+1]$, $f$ maps the vertex $p_i$ in $G$ to the vertex $p_i$ in $H$.
\end{lemma}

\begin{proof}
We see that for $i\in [k-1]$ vertex $p_{i+k+1}$ can map only to vertex $p_{i+k+1}$, since it is connected to vertex $v^{i+1}_1$, none of the vertices in $B_{i}$ are connected to $v^{i+1}_1$ (so if mapped to vertices in $B_{i}$, this mapping would not preserve edges incident to $p_{i+k+1}$), and $p_{i+k+1}$ is also connected to $v^{i}_{l}$. 

Since none of the vertices in $B_{k}$ is connected to $v^{k}_l$, and $p_{2k+1}$ is connected to $v^{k}_{l}$, this is also the reason why $p_{2k+1}$ in $G$ must map to $p_{2k+1}$ in $H$.

The vertex $p_{k+1}$ in $G$ can only be mapped to the vertex $p_{k+1}$ in $H$, since it is connected to neighboring vertices $v^1_1$ on the right and $d^k_{l+k-1}$ on the left.

We see that for $i\in [k]$ vertex $p_{i}$ can map only to vertex $p_{i}$, since it is connected to neighboring vertices $d^{i+1}_1$ on the right and $d^{i}_{l+k-1}$ on the left from $p_{i}$.
\end{proof}


\begin{lemma}
\label{lem:AitoCiNPC}
    Let $G$ be an ordered graph defined in Definition ~\ref{def:GCoreNPC}, and $f:V(G)\to V(H), H\subset G$ be an ordered homomorphism. Then for each $i\in [k]$, $f$ maps all the vertices in $A_i$ to a vertex $v^i_{j_i},j_i\in[l]$ in $C_i$, and all the vertices in $D_i$ to a vertex $d^i_{j'_i},j'_i\in[l+k-1]$ in $D_i$, where $v^i_{j_i}$ and $d^i_{j'_i}$ are connected by collapsing edge.
\end{lemma}

\begin{proof}
We begin the proof by showing where the vertices in $G$ can possibly map.

Lemma ~\ref{lem:pitopiNPC} shows us that all of the $2k+1$ vertices $p_i, i\in[2k+1]$ in $G$ will map by identity to vertices $p_i$ in $H\subset G$.

Let us fix $i\in[k]$. We see that the vertices in $B_i$ can only map to the vertices in $A_i$, since the vertex $p_{k+i+1}$ is connected only to vertices in $C_i$ and $C_{i+1}$, while the vertices in $B_i$ are connected to other vertices in $B_j$, and vertex $p_{k+i}$ is connected to vertex $v^i_1$ by an edge, so vertices in $B_i$ cannot map to $p_{k+i}$, since this mapping would not preserve the order of vertices. 

We also see that vertex $v$ in $B_i$ can only map to the vertices in $C_i$, since other vertices in $B_i$ could not preserve edges incident to $v$, if $v$ mapped to other vertices in $B_i$. Therefore the vertices in $B_i$ can either map to the vertices in $C_i$, or they map to themselves by identity.

The vertices in $C_i$ can only map to the vertices in $C_i$, since no other vertices are connected to $p_{i+k}$ and $p_{i+k+1}$.

Also, vertices in $D_i$ can only map to the vertices in $D_i$, since no other vertices are connected to $p_{i}$ and $p_{i+1}$.

We therefore see that vertices in $A_i$ and $D_i, i\in [k]$ are the only vertices in $G$ that can map to other vertices in $G$, since by Lemma ~\ref{lem:pitopiNPC} vertices $p_{i'},i'\in [2k+1]$ must map by identity.

Next aspect we notice is that since $H\subset G$, at least one vertex in $G$ will map to the subgraph of $G$ by non-identity mapping. That means that in at least one $D_i$ or $A_i$, at least one vertex will map to other vertex than itself.

But from Lemma ~\ref{lem:MCiCollapsible} we know that since every $D_i$ and $A_i, i\in [k]$ form an edge-collapsible ordered graph $M^C_{l+k-1}$  and $l>3$, if any vertex in $D_i$ or $A_i$ maps to some other vertex, then all vertices in $D_i$ must map to just one vertex in $D_i$, and all the vertices in $A_i$ must map to just one vertex in $A_i$, and from the reasoning above, more precisely all the vertices in $A_i$ must map to just one vertex in $C_i$.

Let us now show that if all the vertices in one $B_i$ map to just one vertex in $C_i$, then also all vertices in other $B_j, j\in [k], j\not =i$ must map to just one vertex in $C_j$.

We will show this by contradiction, assuming that all vertices in other $B_j, j\in [k], j\not =i$ will map by identity to an ordered subgraph of $G$.

Note that vertices in other $B_j, j\in [k], j\not =i$ that are adjacent in $G$ to a vertex in $B_i$, are adjacent to a vertex in $B_i$ in $G$ that has degree one. However, if all the vertices in $B_i$ will map to a single vertex $v$ in $C_i$, the vertex $v$ will have degree $k-1$. Therefore, in order for $H$ to be an ordered subgraph of $G$, the vertices in $B_j, j\in [k], j\not =i$ that are adjacent to vertices in $B_i$ in $G$ cannot map by identity, a contradiction.

Since vertices in our $B_i$ are attached to every $B_j, j\in [k], j\not =i$, by the argument presented above (using Lemma ~\ref{lem:MCiCollapsible}), all the vertices in every $B_j, j\in [k], j\not =i$ will need to map to a single vertex in their corresponding sets $C_j, j\in [k], j\not =i$.

But once again by Lemma ~\ref{lem:MCiCollapsible}, for each fixed $j\in [k], j\not =i$, if the vertices in each $B_j$ map to just one vertex in $C_j$, then also all vertices in $D_j$ and all vertices in $A_j$ must map to just one vertex in $D_j$ and one vertex in $A_j$, respectively.

In order for $f$ to preserve edges, it is clear that for every $i\in [k]$, $d^i_{j'_i}, j'_i\in [l+k-1]$ must be such that it is connected to $v^i_{j_i}, j_i\in [l]$ by an edge in $G$, and by construction this is a collapsing edge.
\end{proof}




Now we can continue and finish the proof of Theorem ~\ref{thm:corechiNPC}.

From Lemma ~\ref{lem:AitoCiNPC}, for each $i \in [k]$ there is exactly one $j_i \in [l]$ such that $f(v^i_{j_i})=v^i_{j_i}$.
Putting $C_G = \bigcup_{i \in [k]} \{v^i_{j_i}\}$, we claim that $C_G$ is a clique in $G$.
For contradiction, suppose that $v^i_{j_i}$ is not adjacent to~$v^{i'}_{j_{i'}}; i,i'\in [k]; i\not =i'; j_{i'},j_{i}\in [l] $.
However, there is an edge between the vertices $w^{i}_{i'}\in B_i$, and $w^{i'}_{i}\in B_{i'}$, by the definition of $G$. And $w^{i}_{i'}, w^{i'}_{i}$ are mapped to $v^i_{j_i}, v^{i'}_{j_{i'}}$, respectively; therefore, we obtain a contradiction.

Therefore, from the construction of $G$ based on $F$, we can define $C_F = \bigcup_{i \in [k]} \{x^i_{j_i}\}$, which will be a clique in $F$, intersecting every set $V_i,i\in [k]$.



Since $|V(G)| = 2k+1+2k(l+k-1)=2kl+2k^2+1$ and $|V(H)| = 4k+1$, we see that this is a polynomial reduction; and since again $|V(H)|=\chi^<(G)$, the claimed hardness statements of the \corechi (parameterized) problem follow.
This completes the proof.
\end{proof}

The following result then follows. 

\begin{corollary}
\label{cor:coreNPC}
\core is \NP-complete and \wone-hard when parameterized by the number of the vertices of the image ordered proper subgraph of the instance ordered graph.
        Furthermore, the problem in $n$-vertex instances cannot be solved in time $2^{o(n)}$ nor in time $n^{o(k)}$, unless the ETH fails.
\end{corollary}


\begin{proof}
All hardness statements follow directly from Theorem ~\ref{thm:corechiNPC}.

The remaining point to show is that \core is in \NP. But this is clear since the ordered homomorphism certificate can be verified in polynomial time.
\end{proof}

We note that the \wone-hardness of \corechi may come as surprising since, placed in contrast to the fixed parameter tractability of the problem of finding homomorphisms $G\to H$, for given ordered graphs $G, H$, parameterized by $|H|=\chi^<(G)$ (see ~\cite{certik_complexity_2025}), shows that two seemingly similar parameterized problems might have very different parameterized complexity.

We conclude this article by showing that the \corek and \corechi problems parameterized by $k$ are in \XP.

\begin{obs}
    \corek and \corechi on $n$-vertex instances can be solved in time $n^{\Oh(k)}$. Thus, these problems are in \XP, when parameterized by $k$.
\end{obs}

\begin{proof}
    For \corechi, we can compute $\chi^<(G)=k$ in polynomial time (see ~\cite{certik_complexity_2025}). Then for both problems \corek and \corechi parameterized by $k$, we can exhaustively guess the $k$ vertices of the core (giving us $\Oh(n^k)$ options), and then check in polynomial time if there is a retraction, using Theorem ~\ref{thm:RetInPP}.
\end{proof}

\bibliographystyle{splncs04}
\bibliography{mybibliography}
\end{document}